\documentclass[%
jmp,
 amsmath,amssymb,
reprint, onecolumn 
]{revtex4-1}

\usepackage{graphicx}
\usepackage{dcolumn}
\usepackage{bm}

\usepackage[utf8]{inputenc}
\usepackage[T1]{fontenc}
\usepackage{mathptmx}
\usepackage{etoolbox}

\usepackage{graphicx,color,bbold,bm}
\usepackage{mathrsfs}

\newtheorem{theorem}{Theorem}

\newtheorem{definition}[theorem]{Definition}

\newtheorem{proposition}[theorem]{Proposition}
\newtheorem{remark}[theorem]{Remark}

\newenvironment{proof}[1][Proof]{\noindent\textbf{#1.} }{\ \rule{0.5em}{0.5em}}

\makeatletter
\def\@email#1#2{%
 \endgroup
 \patchcmd{\titleblock@produce}
  {\frontmatter@RRAPformat}
  {\frontmatter@RRAPformat{\produce@RRAP{*#1\href{mailto:#2}{#2}}}\frontmatter@RRAPformat}
  {}{}
}%
\makeatother
\begin{document}

\title{Reproducing Kernel Hilbert Space Approach to Non-Markovian Quantum Stochastic Models}

\author{John E.~Gough} \email{jug@aber.ac.uk}
   \affiliation{Aberystwyth University, SY23 3BZ, Wales, United Kingdom}

\author{Haijin Ding} \email{haijin.ding@centralesupelec.fr}
    \affiliation{Laboratoire des Signaux et Syst\'{e}mes (L2S), CNRS-CentraleSup\'{e}lec-Universit\'{e} Paris-Sud, Universit\'{e} Paris-Saclay, 3, Rue Joliot Curie, 91190, Gif-sur-Yvette, France.}
	
\author{Nina H. Amini} \email{nina.amini@l2s.centralesupelec.fr}
   \affiliation{Laboratoire des Signaux et Syst\'{e}mes (L2S), CNRS-CentraleSup\'{e}lec-Universit\'{e} Paris-Sud, Universit\'{e} Paris-Saclay, 3, Rue Joliot Curie, 91190, Gif-sur-Yvette, France.}

\begin{abstract}
We give a derivation of the non-Markovian quantum state diffusion equation of Di{\'o}si and Strunz starting from a model of a quantum mechanical system coupled to a bosonic bath. We show that the complex trajectories arises as a consequence of using the Bargmann-Segal (complex wave) representation of the bath. In particular, we construct a reproducing kernel Hilbert space for the bath auto-correlation and realize the space of complex trajectories as a Hilbert subspace. The reproducing kernel naturally arises from a feature space where the underlying feature space is the one-particle Hilbert space of the bath quanta. We exploit this to derive the unravelling of the open quantum system dynamics and show equivalence to the equation of Di{\'o}si and Strunz. We also give an explicit expression for the reduced dynamics of a two-level system coupled to the bath via a Jaynes-Cummings interaction and show that this does indeed correspond to an exact solution of the Di{\'o}si-Strunz equation. Finally, we discuss the physical interpretation of the complex trajectories and show that they are intrinsically unobservable.
\end{abstract}

\maketitle

\section{Introduction}
An ensemble of quantum states of a system $S$ is mathematically described as follows: the state of the system at time $t$ is given by $| \Psi_t (\xi )\rangle_S$ where $\xi$ is information describing external influence and where $\xi$ belongs to some collection $\mathcal{E}$ which is a measurable space with a probability measure $\mathbb{P}$. Given an observable $\hat X$, its ensemble average at time $t$ should be
\begin{eqnarray}
\langle \hat X \rangle_t = \int_{\mathcal{E}} \langle \Psi_t (\xi ) | \hat X | \Psi_t (\xi ) \rangle_S \, \mathbb{P} [ d\xi ] .
\end{eqnarray}
It may be that $\langle \hat X \rangle_t $ computed this way yields the correct physical average value of every observable $\hat X$ of the system at any time $t$, in which case we say that the ensemble gives an unraveling of the systems dynamics. However, two points should be noted: firstly, the ensemble is here a \textit{classical} randomization of an otherwise quantum dynamics; secondly, the unraveling only tells us the one-time average. We have not obliged the unravelling to provide the correct multi-dimensional correlations between several observables at different times: only the one-time averages. Kolmogorov's results on encapsulating stochastic processes tell us that unravelings address only a very specific modeling aspect of the noise.

The rationale falls into three (potentially overlapping) camps:
\begin{enumerate}
\item (Quantum State Diffusion)\cite{Pearle,Gisin,GRW,Diosi,GPR} $\xi$ represents some universal external influence which leads to a non-unitary dynamics for the ensemble and may be used to explain the collapse of the wave-function;
\item (Quantum Filtering)\cite{Belavkin,BB,BvHJ}
 $\xi$ is information obtained through indirect (weak) measurement of the system and $\langle \hat X \rangle_t $ is the state of the system at time $t$ conditioned on the information at that time;
\item (Quantum Monte Carlo)\cite{DCM,Carmichael} the ensemble is just a handy device for computing the average behaviour of an open system, with $\xi$ being one simulation out of many which we repeat to get a good estimate.
\end{enumerate}

In each of the three applications above, $\xi$ will typically be a sample path and $\mathcal{E}$ some space of admissible trajectories which we average over. The sample path structure may lead to a filtration with $\xi \mapsto | \Psi_t (\xi ) \rangle_S $ being an adapted process. 
In the quantum filtering approach, the $\xi$ give the continuous-in-time measurement readout from indirect measurements and parallels traditional stochastic filtering methodology - this is often the motivation in the quantum monte carlo approach, even if not made explicit. The 
quantum state diffusion is often more speculative and vague when it comes to the origins of the external noise. See the well-known book of Joos \textit{et al.} on Decoherence for more details\cite{Decoherence}.

Our interest in this paper will be the non-markovian quantum stochastic dynamics introduced by Di\'{o}si and Strunz \cite{DS97} (see also their later work with Gisin\cite{SDG})
\begin{eqnarray}
\frac{d}{d t} | \Psi_t (\zeta ) \rangle_S  =  \bigg( -i \hat{H} +
 \hat{L}  \zeta_t - \hat{L}^\ast \int_0^t d\tau \,K\left( t,\tau \right) \frac{\delta }{\delta \zeta _\tau } \bigg)  | \Psi _t(\zeta ) \rangle_S ,
\label{eq:nMSSE}
\end{eqnarray}
where $\zeta$ is a complex-valued trajectory. Here $K(t,s)$ is the auto-correlation function of the bath $B$ which acts as environment for $S$. 

Complex trajectories originally appeared in the (markovian) quantum state diffusion approach where they were motivated by the internal symmetries of the Lindblad generator: here the $\zeta$ appear as the sample paths of a complex Wiener process $Z(t) = Z^\prime (t) +i Z^{\prime \prime} (t)$ where $Z^\prime$ and $Z^{\prime \prime}$ are independent standard Wiener processes. However, they are atypical of the other approaches where each measurement process is a real-valued process. However, they are compatible with quantum filtering\cite{PUD2017,JG2017}.

In the paper, we show that the equation (\ref{eq:nMSSE}) can be derived from a microscopic model. The origin of the complex trajectories ultimately comes from the fact that we are using the Bargmann-Segal (complex-wave) representation for the bath. Here our approach follows the derivation given in Di\'{o}si and Strunz \cite{DS97}, but we deal with technical mathematical issues relating to infinite baths - we also explain the origin of the causal structure of (\ref{eq:nMSSE}): the kernel $K(t,s)$ is not causal and this is not built in to the Bargmann-Segal representation, but instead comes from the form of the dynamical equations we wish model. In setting up a rigorous description, we make use of the fact that the kernel leads to a reproducing kernel Hilbert space (RKHS) and we show that space of complex trajectories is a Hilbert subspace of this RKHS. In particular, the mircroscopic model naturally supplies a feature map for the kernel.

As an application, in Section \ref{sec:Exact} we construct an explicit unravelling for the problem of a two-level atom coupled to the bath via Jaynes-Cummings interaction, and then show that the state $| \Psi_t (\zeta ) \rangle_S$ does indeed satisfy (\ref{eq:nMSSE}).

In section \ref{sec:observe}, we show that the complex trajectories have no physical interpretation as observable processes. They are entirely a mathematical by-product of the Bargmann-Segal representation.

Care is needed, therefore, in the interpretation. For instance, the equation (\ref{eq:nMSSE}) is often compared with the quantum Zakai equation\cite{SDG}. Both of these are indeed linear stochastic equations for an un-normalized vector state, however, only the latter has an interpretation as the conditional evolution of a continuously monitored system. 

\section{The Microscopic Model}
We consider a system with Hilbert space $\mathfrak{h}_S$ coupled to a bosonic bath with Hilbert space $\mathfrak{H}_B = \Gamma (\mathfrak{f}_B)$, that is, the (Bose) Fock space over a fixed one-particle space $\mathfrak{f}_B$. 

\subsection{The Bath Fock space}
Recall that $\mathfrak{H}_B = \Gamma (\mathfrak{f}_B)$ will be the direct sum of the symmetrized Hilbert spaces $\otimes_{\mathrm{symm.}}^n \mathfrak{f}_B$ for $n=0,1,2, \cdots$, with the $n=0$ space being spanned by a single vector $|\mathrm{vac} \rangle_B$ which is the Fock vacuum.
(See, for instance, Chapter 10 of Gough and Kupsch \cite{GK} for details.)

The creation and annihilation operators for test-function $g\in \mathfrak{f}_B$ are denoted as $\hat{a}(g)^\ast, \hat{a}(g)$ respectively, and the differential second-quantization of an observable $\hat{M}$ by $d\Gamma (\hat{M})$. 

\bigskip

For example, a bath consisting of a discrete set of oscillators with non-degenerate frequencies $\Omega$ is represented by taking $\mathfrak{f}_B = \mathbb{C}^\Omega$ with inner-product $\langle f | g\rangle = \sum_{\omega \in \Omega } f (\omega )^\ast g (\omega )$ and
\begin{eqnarray}
\hat{a}(g)^\ast = \sum_{\omega \in \Omega} g (\omega ) a_\omega^\ast, 
\quad \hat{a}(g) = \sum_{\omega \in \Omega} g (\omega )^\ast a_\omega,  \quad
d\Gamma ( \hat{M}) = \sum_{\omega , \omega' \in \Omega} M_{\omega, \omega'} a^\ast_\omega a_{\omega'} ,
\end{eqnarray}
where $[\hat{a}_\omega , \hat{a}_{\omega'}^\ast ] = \delta_{\omega , \omega'}$ and $\hat{M} $ is a matrix with entries $M_{\omega , \omega '}$.
The bath Hamiltonian would then be $\hat{H}_B = d \Gamma (\hat{h}_B ) = \sum_\omega \omega \hat{a}^\ast_\omega \hat{a}_\omega$, with $\hat{h}_B$ being the diagonal matrix with entries $\omega \in \Omega$.

For a continuum of oscillators, we could take $\mathfrak{f}_B = L^2( \mathbb{R}_+ , d \omega )$ with $\hat{a}(g)^\ast = \int_0^\infty g (\omega ) \hat{a}_\omega^\ast\, d\omega , \, \hat{a}(g) = \int_0^\infty g(\omega)^\ast \hat{a}_\omega \,  d\omega$ where now $[ \hat{a}_\omega , \hat{a}_{\omega'}^\ast ] = \delta (\omega - \omega' )$.

\subsection{Bargmann-Segal (Complex Wave) Representation}
For a given vector $| f\rangle \in \mathfrak{f}_B$, its \textit{exponential vector} is defined as
\begin{eqnarray}
| e^f \rangle_B = e^{ \hat{a}(f) } | \mathrm{vac} \rangle_B \equiv
| \mathrm{vac} \rangle \oplus | f \rangle \oplus \big( \frac{1}{\sqrt{2!}} | f \rangle \otimes | f \rangle \big) \oplus \cdots .
\end{eqnarray}
In particular the vacuum state corresponds to $ | \mathrm{vac} \rangle_B \equiv| e^f \rangle_B $. The exponential vectors form a total subset in $\mathfrak{H}_B$ and we note that
\begin{eqnarray}
\langle e^f | e^ g \rangle_B = e^{ \langle f | g \rangle } .
\end{eqnarray}

The Complex Wave (or Bargmann-Segal) representation of a vector $|\Psi \rangle \in \mathfrak{H}_B$ is in terms of the function $\tilde{\Psi}$ given by
\begin{eqnarray}
\tilde{\Psi}\left( f \right) =\langle e^f \,|\Psi
\rangle .
\end{eqnarray}
The mapping $f \mapsto \langle e^f |\Psi \rangle $ is naturally anti-holomorphic.  

\begin{definition}
The pre-measure $\tilde{\mathbb{P}}$ on the one-particle space $\mathfrak{f}_B$ by
\begin{eqnarray}
\int_{\mathfrak{f}_B} e^{ \langle g_1 | f \rangle}  \, e^{ \langle f | g_2 \rangle} \, \tilde{\mathbb{P}}
\left[ d f \right] = e^{ \langle g_1 | g_2 \rangle} .
\label{eq:inner}
\end{eqnarray}
\end{definition}

We remark that for a finite assembly of oscillators the one-particle Hilbert space is $\mathfrak{f}_B \cong \mathbb{C}^\Omega$ which can be considered as the measurable space $\mathbb{R}^{2 \Omega}$. In this case, $\mathbb{P}$ is a probability measure and is given explicitly by
\begin{eqnarray}
\tilde{\mathbb{P}}\left[ df \right] =\prod_{\omega }\left(
\frac{1}{\pi} e^{-|f  (\omega ) |^{2}}df (\omega )^{\prime } df (\omega )^{\prime \prime } \right) 
\end{eqnarray}
where we write each variable $f (\omega )$ in terms of its real and imaginary parts as $ f (\omega )= f (\omega )^{\prime }+i f ( \omega )^{\prime \prime }$. 

As a technical aside, $\tilde{\mathbb{P}}$ will define a Gaussian measure only for the finite dimensional case. In the infinite-dimensional case, we should extend $\tilde{\mathbb{P}}$ to a $\sigma$-additive measure over a larger space $\mathfrak{f}_B^>$. As before, for each $| \Psi \rangle \in \mathfrak{H}_B$ we arrive at functions $: f \mapsto \tilde{\Psi} (f  )$ this time defined for $\alpha \in \mathfrak{f}_B^>$. These functions are properly square-integrable with respect to the measure (which we again denote as $\tilde{\mathbb{P}}$), however, one a dense subset of these functions will now be continuous and anti-holomorphic. See Gough and Kupsch \cite{GK} for details.

In general, we have 
\begin{eqnarray}
\langle \Phi |\Psi \rangle_B =\int_{\mathfrak{f}_B } \tilde{\Phi}(f) \tilde{\Psi} (f)\, \tilde{\mathbb{P}}\left[ df \right] ,
\end{eqnarray}
with (\ref{eq:inner}) being the special case $\Phi = e^{g_1}$ and $\Psi = e^{g_2}$. We also have the resolution of identity
\begin{eqnarray}
\int_{\mathfrak{f}_B }|e^f \rangle \langle e^f | \,\, \tilde{\mathbb{P}} \left[ df \right] = \hat{I}_B .
\end{eqnarray}

\subsection{The System-Bath Interaction}
\label{sec:SB}

The total Hamiltonian for our system and bath will be of the form
\begin{eqnarray}
\hat{H}_{\mathrm{Tot.}} = \hat{H} \otimes \hat{I}_B + \hat{I}_S \otimes \hat{H}_B +i\hat{L} \otimes \hat{a}(g)^\ast -i\hat{L}^\ast \otimes \hat{a}(g) .
\label{eq:Ham_SB}
\end{eqnarray}
where $\hat{H}=\hat{H}^\ast, \hat{L}$ are bounded operators of the system, $| g\rangle  \in \mathfrak{f}_B$ and $\hat{H}_B = d \Gamma (\hat{h}_B)$ for some self-adjoint operator $\hat{h}_B$ on $\mathfrak{f}_B$. We note that the annihilation operators evolve under the bath Hamiltonian according to
\begin{eqnarray}
e^{it\hat{H}_B} \hat{a}(g) e^{-it\hat{H}_B} \equiv \hat{a} (g_t), \, \text{where} \, |g_t \rangle = e^{it\hat{h}_B} | g \rangle  .
\end{eqnarray}
The bath correlation function is then given by
\begin{eqnarray}
K(t,s) \triangleq  \langle g_t | g_s \rangle = \langle g | e^{-i (t-s) \hat{h}_B} g \rangle .
\label{eq:K}
\end{eqnarray}

We note that the kernel is sesquilinear, i.e.,  $K\left( t,s\right) ^{\ast }=K\left( s,t\right) $, and defines a positive definite kernel. We note that kernels of the form (\ref{eq:K}) cannot be both non-trivial and \textit{causal}. Causality is the condition that $K(t , \tau) \equiv 0$ for $ \tau <t$, but by the sesquilinear symmetry we would also have $K(t , \tau ) =0$ for $ t < \tau$ too. The only nontrivial and causal model is the memoriless kernel $K(t,s) = k_0 \, \delta (t-s)$.

\begin{definition}
The vector $g\in \mathfrak{f}_B$ is said to be faithful if the measure $\mu_g$, with Fourier transform $\int_{-\infty}^{\infty} e^{i\omega t}\,
\mu_g [d \omega] = \langle g |e^{it\hat{h}_B} | g \rangle$, has support equal to the spectrum of $\hat{h}_B$.
\end{definition}
This basically means that all the bath modes enter in the interaction. In general, $\mathfrak{h}_B$ mat be written as the direct sum of a maximal Hilbert space for which $g$ is faithful (the modes that couple to the system) and its orthogonal complement. By ignoring this complement we can always assume our setup is faithful.

\begin{remark}
For instance, the kernel for a discrete assembly of oscillators will be $K(t,s) = \sum_{\omega \in \Omega } |g (\omega ) |^2 e^{-i(t-s) \omega}$ and faithfulness would just mean that $g(\omega ) \neq 0$ for all $\omega\in \Omega$.
\end{remark}

\subsection{The Open System Dynamics}

It is convenient to move to the interaction picture using the unitary family $U_t = e^{+it H_0} e^{-itH_{\mathrm{Tot.}}}$ where $\hat{H}_0 = \hat{I}_S \otimes \hat{H}_B$ and here we find the differential equation
\begin{eqnarray}
\frac{d}{dt}  \hat{U} _t = - i \hat{\Upsilon}_t \, \hat{U}_t, \qquad -i\hat{\Upsilon }_t = -i\hat{H}\otimes \hat{I }_B + \hat{L} \otimes \hat{Z}(t) -\hat{L}^\ast \otimes \hat{Z}(t)^\ast ,
\end{eqnarray}
where $\hat{Z}(t) = \hat{a}(g_t)^\ast$. It is important to note that $\hat{Z} (t)$ is not Hermitean and therefore not an observable! Indeed, the canonical commutation relations imply that
\begin{eqnarray}
[\hat{Z}(t)^\ast , \hat{Z}(s) ] = K(t,s) \, \hat{I}_B.
\label{eq:Z_CR}
\end{eqnarray}
(The reason for taking $\hat{Z}(t)$ to be the creator rather than the annihilator is that we will be shortly using the anti-holomorphic Bargmann representation.) We note the eigen-relation
\begin{eqnarray}
\bigg( \hat{Z}^\ast_t - \langle g_t | f \rangle\bigg)  \, |e^f \rangle _B = 0.
\label{eq:eigen}
\end{eqnarray}

The two-point vacuum expectations are
\begin{gather}
\langle \mathrm{vac} | \hat{Z}_t^\ast \hat{Z}_s | \mathrm{vac} \rangle = K\left( t,s\right) , \quad
\langle \mathrm{vac} | \hat{Z}_t \hat{Z}_s | \mathrm{vac} \rangle = \langle \mathrm{vac} | \hat{Z}_t^\ast \hat{Z}_s^\ast | \mathrm{vac} \rangle =\langle \mathrm{vac} | \hat{Z}_t \hat{Z}_s^\ast | \mathrm{vac} \rangle =0.
\label{eq:two-point_vac}
\end{gather}

It is convenient to introduce the following operators
\begin{eqnarray}
j_t (\hat{X}) &=& \hat{U}_t^\ast ( \hat{X} \otimes \hat{I}_B) \hat{U}_t ,\\
\hat{Z}_{\mathrm{in}} (t) &=& \hat{I}_S \otimes \hat{Z}(t) , \\
\hat{Z}_{\mathrm{out}} (t) &=& \hat{U}_t^\ast \hat{Z}_{\mathrm{in}} (t) \hat{U}_t.
\end{eqnarray}

\begin{proposition}[Equal time commutation relations]
For each system operator $\hat{X}$, the operator $j_t (\hat{X})$ commutes with both $\hat{Z}_{\mathrm{out}} (t)$ and $\hat{Z}_{\mathrm{out}} (t)^\ast$. 
\end{proposition}
To see this, we note that, for instance,
\begin{eqnarray}
[j_t (\hat{X}), \hat{Z}_{\mathrm{out}} (t) ] = \hat{U}_t^\ast [ \hat{X} \otimes \hat{I}_B , \hat{I}_S \otimes \hat{Z} (t) ] \hat{U}_t =0, \, \mathrm{etc.}
\end{eqnarray}

We should emphasize that the outputs $\hat{Z}_{\mathrm{out}} (t)$ and $\hat{Z}_{\mathrm{out}} (t)^\ast$ do not commute with $j_s (\hat{X})$ for $s \neq t$! Moreover, while the inputs satisfy $ [\hat{Z}_{\mathrm{in}}(t)^\ast , 
\hat{Z}_{\mathrm{in}} (s) ] = K(t,s) \hat{I}$, where $\hat{I}= \hat{I}_S \otimes \hat{I_B}$, there is no such simple relation for the outputs.

\begin{proposition}[Input-Output Relation]
The processes $\hat{Z}_{\mathrm{in}} (t)$ and $\hat{Z}_{\mathrm{out}} (t)$ are related by the input-output relations
\begin{eqnarray}
 \hat{Z}_{\mathrm{out}} (t)^\ast = \hat{Z}_{\mathrm{in}} (t)^\ast + \int_0^t K(t, \tau ) j_\tau (\hat{L}) d\tau .
\label{eq:io}
\end{eqnarray}
\end{proposition}

\begin{proof}
Fix $f \in \mathfrak{f}_B$, then $\frac{d}{dt} \, \hat{U}_t^\ast ( \hat{I}_S \otimes \hat{a}(f) \big) \hat{U}_t = \langle f | g_t \rangle \, \hat{U}_t^\ast \big( \hat{L} \otimes \hat{I}_B \big) U_t = \langle f | g_t \rangle \, j_t (\hat{L}) $,
and integrating leads to
\begin{eqnarray}
 \hat{U}_t^\ast ( \hat{I}_S \otimes \hat{a}(f) \big) \hat{U}_t = \hat{I}_S \otimes \hat{a}(f) + \int_0^t \langle f | g_\tau \rangle  \, j_\tau (\hat{L}) \, d \tau .
\end{eqnarray}
Setting $f=g_t$ gives the desired result.
\end{proof}

As a corollary, we obtain the following version of the Ehrenfest equations for system operators.
\begin{proposition}[Ehrenfest equation]
We fix the initial state to be $ |\phi \rangle_S \otimes | \mathrm{vac} \rangle_B $ and set $\langle j_t (\hat{X}) \rangle = \langle \phi    \otimes  \mathrm{vac} | j_t (\hat{X} ) \, \phi   \otimes  \mathrm{vac} 
\rangle$. Then
\begin{eqnarray}
\frac{d}{dt} \langle  j_t (\hat{X}) \rangle = \langle j_t ( \frac{1}{i} [\hat{X},\hat{H}] ) \rangle 
+ \int_0^t d\tau \, K(t, \tau ) \langle j_t ([\hat{L}^\ast , \hat{X} ] ) j_\tau (\hat{L}) \rangle
+ \int_0^t d\tau \, K(t, \tau )^\ast \langle j_\tau (\hat{L}^\ast ) j_t ( [\hat{X},\hat{L}]) \rangle .\nonumber 
\end{eqnarray}
\end{proposition}
\begin{proof}
We start with the equation of motion $\frac{d}{dt} j_t (\hat{X}) = \hat{U}_t^\ast \frac{1}{i} [ \hat{X} \otimes \hat{I}_B, \hat{\Upsilon }_t ] \hat{U}_t$, which implies
\begin{eqnarray}
\frac{d}{dt} j_t (\hat{X})  = j_t ( \frac{1}{i} [\hat{X},\hat{H}] ) + j_t ( [\hat{X},\hat{L}]) \hat{Z}_{\mathrm{out}} (t) + j_t ([\hat{L}^\ast , \hat{X} ] ) \hat{Z}_{\mathrm{out}} (t)^\ast .
\end{eqnarray}
Note that the equal time commutation relations allow us to commute either $\hat{Z}_{\mathrm{out}} (t)$ or $\hat{Z}_{\mathrm{out}} (t)^\ast$ with the terms multiplying them. The desired equation is obtained by replacing the output processes with the inputs using (\ref{eq:io}) and using the fact that $\hat{Z}_{\mathrm{in}} (t)^\ast \, | \phi   \otimes  \mathrm{vac}  \rangle \equiv 0$ since $\hat{Z}_{\mathrm{in}} (t)^\ast$ is an annihilator.
\end{proof}

\subsection{Caldeira-Leggett Baths}
We take a continuum of bath modes. The bath Hamiltonian is given by
\begin{eqnarray}
\hat{H}_B = \frac{1}{2} \int_0^\infty d \omega \bigg\{ \frac{\omega}{ J (\omega )} \hat{\pi} (\omega ) ^2 + \omega J (\omega ) \hat{q} (\omega )^2 \bigg\}
\end{eqnarray}
where $\hat{q} (\omega) $ and $\hat{\pi} (\omega ) $ satisfy canonical commutation relations $[ \hat{q}(\omega ) , \hat{\pi} (\omega ' ) ] = i  \, \delta (\omega - \omega' )$. Here $J$ is a function known as the $J$ \textit{spectral density} function.  We introduce the mode operators as
\begin{eqnarray}
\hat{a}_\omega   = \sqrt{ \frac{J(\omega ) }{2 }} \hat{q} ( \omega ) + i  \frac{1}{ \sqrt{ 2  J (\omega )} } \hat{\pi} (\omega ) 
\end{eqnarray}
in which case $\hat{H}_B =  \int_0^\infty   \omega \, \hat{a}(\omega )^\ast \hat{a} (\omega ) \, d \omega $.

We now couple our system $S$ coupled to the bath through its position coordinate $\hat{q}$: the total Hamiltonian will be taken to be of Caldeira-Leggett form \cite{CL}
\begin{eqnarray}
\hat{H}_{S+B} &=& \hat{H} \otimes \hat{I}_B +  \frac{1}{2} \int_0^\infty d \omega \bigg\{ \frac{\omega}{ J (\omega )} \big[ \hat{I}_S \otimes \hat{\pi} (\omega ) 
- \sqrt{ \frac{2}{\pi}} \hat{q} \otimes \hat{I}_B \big]^2   + \omega J (\omega ) \hat{I}_S \otimes \hat{q} (\omega )^2 \bigg\} \nonumber \\
&=& \hat{H} \otimes \hat{I}_B + \hat{I}_S \otimes \hat{H}_B + \hat{V}_{SB} .
\end{eqnarray}
where the interaction is $\hat{V}_{SB} =i \hat{q} \otimes  \int_0^\infty d \omega  \sqrt{ \frac{   J (\omega )}{\pi} }\big( \hat{a}_\omega - \hat{a}_\omega ^\ast \big) $.

This can be put in the standard form (\ref{eq:Ham_SB}) with $\hat L = -i \hat q$ and $ g ( \omega ) =   \sqrt{   J(\omega ) /\pi } $.
At this stage, we identify the complex wave representation. Our one particle Hilbert space is $\mathfrak{f}_B = L^2 ( \mathbb{R}_+, d \omega )$ and for every $f = f (\omega )$ in we associate a complex trajectory $\zeta (t) = \int_0^\infty \sqrt{   J(\omega ) /\pi } e^{i\omega t} f (\omega )^\ast \, d \omega $. The corresponding quantum process is
\begin{eqnarray}
\hat{Z} (t) = \int_0^\infty  \sqrt{   J(\omega ) /\pi } e^{i\omega t} a _\omega ^\ast \, d \omega ,
\end{eqnarray}
and the kernel will be the Fourier transform of the $J$ spectral density:
\begin{eqnarray}
K (t,s)  =\frac{1 }{\pi} \int_0^\infty J (\omega )  e^{-i\omega (t-s)}  \, d \omega.
\end{eqnarray}


\section{Complex Trajectories}
The eigen-relation (\ref{eq:eigen}) may be alternatively written as
\begin{eqnarray}
\bigg( \hat{Z}^\ast_t - \zeta_t (f)^\ast \bigg)  \, |e^f \rangle _B = 0.
\label{eq:eigen_zeta}
\end{eqnarray}
and we refer to $t \to \zeta_t (f)$ is the \textit{complex trajectory} associated with $f\in \mathfrak{f}_B$.

\subsection{Reproducing Kernel Hilbert Space Formalism}
In the following, $\mathbb{T}$ will denote an interval in $\mathbb{R}$. We fix a positive definite kernel $K: \mathbb{T} \times \mathbb{T} \mapsto \mathbb{C}$. For each $t\in \mathbb{T}$, a function $\mathbb{k}_t: \mathbb{T} \mapsto \mathbb{C}$ is then defined by
\begin{eqnarray}
\mathbb{k}_t (\cdot ) \triangleq K( t, \cdot ) .
\end{eqnarray}

We recall that a Hilbert space $\mathscr{H}$ of complex-valued functions on $\mathbb{T}$ forms a \textit{reproducing kernel Hilbert space (RKHS)} for the kernel $K$ if $\mathbb{k}_t \in \mathscr{H}$ for each $t \in \mathbb{T}$ and we have the \textit{reproducing property}\cite{RKHS}
\begin{eqnarray}
\langle \mathbb{k}_t , f \rangle _{\mathscr{H}} = f(t)
\end{eqnarray}
for all $t \in \mathbb{T}$ and all $f \in \mathscr{H}$.

In this situation, the maps $\mathbb{k}_t$ are called \textit{representers, or evaluation maps,} and they play the role of Dirac delta-functions for the test functions $\mathscr{H}$. Crucially the representers are bounded maps in the RKHS setting. We also note the identity
\begin{eqnarray}
\langle \mathbb{k}_t , \mathbb{k}_s  \rangle _{\mathscr{H}} = \mathbb{k}_s(t) = K(s,t) .
\end{eqnarray}

A kernel is said to be derivable from a \textit{feature map}\cite{RKHS} if there exists a Hilbert space $\mathfrak{f}$ (called the feature space) and a function $g : \mathbb{T} \mapsto \mathbb{C}$ (called the feature map) such that
\begin{eqnarray}
K(t,s)  \equiv \langle g_t , g_s \rangle_{\mathfrak{f}} .
\end{eqnarray}

\begin{remark}
The bath correlation kernel introduced in (\ref{eq:K}) is clearly derivable from a feature map with $\mathfrak{f}_B$ as feature space. The set $\mathbb{T}$ may generally be taken to be some interval of the time in $\mathbb{R}$ and we shall fix $dt$ as Lebesgue measure. In order to construct an RKHS we need further assumptions (Mercer conditions) which essentially allow us to write the kernel in the form
\begin{eqnarray}
K (t,s) = \sum_n \lambda_n \phi_n (t) \phi_n (s)^\ast
\label{eq:Mercer}
\end{eqnarray}
where the $\lambda_n$ are positive numbers and the $\phi_n$ form a complete orthonormal basis for $L^2( \mathbb{T} , dt)$.
Here we require that the integral operator $T_K $ given by $ \big( T_K  u \big) (t) = \int_{\mathbb{T}} K\big( t, s \big) u(s) \, ds$ is Hilbert-Schmidt in which case the $\lambda_n$ are the eigenvalues and the $\phi_n$ are the eigen-functions normalized in $L^2( \mathbb{T} , dt)$. Here one solves the homogeneous Fredholm integral equation
\begin{eqnarray}
\int_{\mathbb{T}} K(t , s ) \phi_n (d) \, ds = \lambda_n \, \phi_n (t).
\end{eqnarray}
This is automatically satisfied when $\mathbb{T}$ is compact but otherwise will require that the kernel has a finite-trace property.\cite{RKHS,DUV}
\end{remark}

With the assumption that the kernel admits a Mercer expansion (\ref{eq:Mercer}), we may construct the RKHS $\mathscr{H}_K (\mathbb{T} ,dt)$ - which we write as $\mathscr{H}$ as shorthand when no confusion arise - as follows: (step 1) we define $\mathscr{H}$ to be the set of all functions $u$ of the form $u (t ) = \sum_n u_n \, \phi_n ( t )$ with $ \sum_n \frac{1}{\lambda_n} | u_n|^2 < \infty$ and we endow this with the inner-product
\begin{eqnarray}
\langle u , v \rangle_{\mathscr{H}} \triangleq
\sum_n \frac{1}{\lambda_n} u_n^\ast v_n ;
\end{eqnarray}
(step 2) we see that the representers take the form
\begin{eqnarray}
\mathbb{k}_t (\cdot ) = \sum_n \lambda_n \phi_n (t)^\ast \, \phi_n (\cdot ) ,
\end{eqnarray}
so that the coefficients of $\mathbb{k}_t $ are $\lambda_n \phi_n (t)^\ast$;
(step 3) for $u \in \mathscr{H}$ we obtain the reproducing property
\begin{eqnarray}
\langle \mathbb{k}_t ,u \rangle_{\mathscr{H}} = \sum_n \frac{1}{\lambda_n} \big( \lambda_n \phi_n (t)^\ast \big)^\ast u_n
=\sum_n u_n \phi_n (t) \equiv u(t).
\end{eqnarray}

As a check, we note that
\begin{eqnarray}
\langle \mathbb{k}_t , \mathbb{k}_s  \rangle_{\mathscr{H}} = \sum_n \frac{1}{\lambda_n} \big( \lambda_n \phi_n (t)^\ast \big)^\ast \lambda_n \phi_n (s)^\ast
=\sum_n \lambda_n \phi_n (t) \phi_n (s)^\ast \equiv K(t,s).
\end{eqnarray}

Note that the \lq\lq delta-function\rq\rq\, on $\mathbb{T}$ is given by
\begin{eqnarray}
\delta_{\mathbb{T}} (t-s) \equiv \sum_n   \phi_n (t) \phi_n (s)^\ast .
\end{eqnarray}

\subsection{Hilbert Space of Complex Trajectories}

\begin{definition}
\label{def:complex_traj}
For each $f\in \mathfrak{f}$, we define its associated complex trajectory over the time interval $\mathbb{T}$ to be the function $\zeta (f): \mathbb{T} \mapsto \mathbb{C}: t \mapsto \zeta_t (f)$ where
\begin{eqnarray}
\zeta_t (f) = \langle f | g_t \rangle_{\mathfrak{f} }.
\end{eqnarray} 
The space of complex trajectories will be denoted as $\mathscr{C}_K (\mathbb{T} , dt )$.
\end{definition}
We note that $\zeta (f)$ is defined to be anti-holonomic in $f$ and this is the natural choice for the Bargmann representation which we will use in the next section. Unfortunately, it makes the following result slightly awkward to state.

\begin{theorem}
The set of complex trajectories, $\mathscr{C}_K (\mathbb{T} , dt )$, forms a Hilbert subspace of the RKHS $\mathscr{H}_K (\mathbb{T} ,dt)$ (inheriting the same inner product) and the map $\zeta $ is conjugate-linear isometry into the feature space $\mathfrak{f}_B$. 
\begin{eqnarray}
\langle \zeta (f_1)^\ast , \zeta (f_2)^\ast \rangle_{\mathscr{H}}
= \langle f_1 |f_2 \rangle_{\mathfrak{f}} .
\end{eqnarray}
The space $\mathscr{C}_K (\mathbb{T} , dt )$ is conjugate-isomorphic to $\mathfrak{f}_B$ when the set up is faithful.
\end{theorem}

\begin{proof}
Let us introduce the following vectors in the feature space:
\begin{eqnarray}
| \pi_n \rangle_{\mathfrak{f}} \triangleq \int_{\mathbb{T}} \phi_n (t) | g_t \rangle_{\mathfrak{f}} \, dt ,
\end{eqnarray}
where the integral is understood as a Gel'fand-Pettis integral. We see that
\begin{eqnarray}
\langle \pi_n | \pi_m \rangle_{\mathfrak{f}} =
\int_{\mathbb{T}} dt \int_{\mathbb{T}} ds \, \phi_n (t) ^\ast \phi_n (s) \, \langle  g_t | g_s \rangle_{\mathfrak{f}} =
\int_{\mathbb{T}} dt \int_{\mathbb{T}} ds \, \phi_n (t) ^\ast K( t, s) \phi_n (s)  = \lambda_n \, \delta_{nm} .
\end{eqnarray}
In particular, this leads to the resolution of identity $ \sum_n \frac{1}{\lambda_n}| \pi_n \rangle  \langle \pi_n |  = I_{\mathfrak{f}}$.

The complex trajectories then have the expansion
\begin{eqnarray}
\zeta_t (f)^\ast = \sum_n \langle \phi_n , \zeta (f)^\ast \rangle_{L^2} \, \phi_n (t) = \sum_n \bigg( \int_{\mathbb{T}}  \phi_n (s)^\ast  \langle g_t | f  \rangle_{\mathfrak{f}} \bigg)  \phi_n (t) = \sum_n \langle \pi_n | f  \rangle_{\mathfrak{f}} \, \phi_n (t) 
\end{eqnarray}
and therefore
\begin{eqnarray}
\langle \zeta (f_1)^\ast , \zeta (f_2)^\ast \rangle_{\mathscr{H}}
= \sum_n \frac{1}{\lambda_n}  \langle f_1 | \pi_n  \rangle_{\mathfrak{f}} \langle \pi_n | f_2  \rangle_{\mathfrak{f}}
= \langle f_1 |f_2 \rangle_{\mathfrak{f}} .
\end{eqnarray}
\end{proof}

Generally speaking, the map $\zeta$ from the feature space to the RKHS may be many-to-one.


\subsection{Quantum Karhunen-Lo\`{e}ve Theorem}
Under the condition that the kernel is a Mercer kernel, the inputs can be given the expansion
\begin{eqnarray}
\hat{Z}_t ^\ast = \sum_n \hat{z}_n^\ast  \phi_n (t) ,
\end{eqnarray}
where the coefficient operators satisfy the commutation relations
\begin{eqnarray}
[ \hat{z}_n^\ast  , \hat{z}_m ] = \lambda_n \, \delta_{nm} .
\end{eqnarray}
Specifically, these are given $\hat{z}_n = \int_{\mathbb{T}} \phi_n (t) \hat{Z}_t \, dt$. One checks that 
\begin{eqnarray*}
[ \hat{z}_n^\ast  , \hat{z}_m ] =
\int_{\mathbb{T}} dt\int_{\mathbb{T}}  ds \,  \phi_n (t)^\ast  [\hat{Z}_t^\ast , \hat{Z}_s] \phi_n (s)= 
\int_{\mathbb{T}} dt\int_{\mathbb{T}}  ds \,  \phi_n (t)^\ast  \, K(t,s) \phi_n (s)= \lambda_n \, \delta_{nm} .
\end{eqnarray*}

The convergence will be uniform in $t$ in the $L^2$-norm. This is essentially a quantum analogue of the classical Karhunen-Lo\`{e}ve Theorem.
It is natural to introduce the bath annihilators
\begin{eqnarray}
\hat{b}_n \triangleq \frac{1}{\sqrt{\lambda}_n} \hat{z}_n^\ast ,
\end{eqnarray}
and we have the canonical commutation relations $[ \hat{b}_n , \hat{b}_m^\ast]  = \delta_{nm} $.

The interaction picture Hamiltonian $\hat{\Upsilon}_t $ may then be written as
\begin{eqnarray}
 -i\hat{\Upsilon }_t = -i\hat{H}\otimes \hat{I }_B + \sum_n \sqrt{\lambda_n} 
\bigg\{ \hat{L} \otimes \hat{b}_n ^\ast \, \phi_n(t)^\ast  -\hat{L}^\ast \otimes \hat{b}_n \, \phi_n (t) \bigg\} ,
\end{eqnarray}

\begin{remark}
What the Karhunen-Lo\`{e}ve expansion does is to decorrelate the process $\hat Z_t$. Specifically, the auto-correlation was described by the kernel $K(t,s)$ however this is replaced by the modes $\hat b_n$ which are independent oscillators in their ground state. Note that for a finite bath, we go from a finite number of modes $\hat a_\omega $ to an infinite number of the $\hat b_n$. 
\end{remark}
\begin{remark}
The annihilation and creation processes $\hat A_t, \hat A_t^\ast$ of the Hudson-Parthasarathy theory satisfy $ [ \hat A_t , \hat A_s^\ast ] = t\wedge s$. restricting to the interval $\mathbb{T} =[0,T]$, the kernel $K(t,s) = t \wedge s$ has eigen-functions $\phi_n (t) = \sqrt{\frac{2}{T}}
\sin ( \omega_n t )$ for $n=1,2,3,\cdots$, where $\omega_n = \frac{\pi}{T} (n - \frac{1}{2} )$, and corresponding eigenvalues $\lambda_n = 1/ \omega_n^2$. We therefore find that
\begin{eqnarray}
\hat A_t \equiv \sum_{n=1}^\infty
\frac{ \sqrt{2T}}{\pi ( n - \frac{1}{2} )} \sin \bigg( \frac{\pi (n- \frac{1}{2} )t}{T} \bigg) \, \hat b_n .
\end{eqnarray}
This expansion was essentially given in \cite{QKL}. Note however that, for the markovian dynamics, the kernel is a $\delta$-function and the driving input processes are quantum white noises, that is, the formal derivatives of these processes, see section \ref{sec:markov}.
\end{remark}


\subsection{Causal Form}
Let us set the time interval to be $\mathbb{T} = [0,T]$ where $T$ is a finite time horizon. We shall write the eigenvalues as $\lambda_n (T)$ and the eigen-functions as $\phi_n (\cdot , T)$ to emphasize the dependence on $T$. We then have the expansion
\begin{eqnarray}
\hat Z^\ast_t \equiv \sum_n \sqrt{\lambda_n (T)} \phi_n (t; T) \, \hat b_n (T) ,
\end{eqnarray}
for each $0 \le t \le T$.
One of the problems with this expansion is that $z_n^\ast (T) = \lambda_n (T) \, \hat b_n (T)$ will be given by $\int_0^T \phi_n (s) \hat Z_s \, ds$, so we are effectively computing $\hat Z_t^\ast$ using both its past values $0 \le s < t$ and its future values $s < t \le T$.

This may be avoided by using a running time horizon (set equal to the time parameter $t$) rather than a fixed future time $T$. We therefore have
the \textit{causal form}
\begin{eqnarray}
\hat Z^\ast_t \equiv \sum_n \sqrt{\lambda_n (t)} \phi_n (t; t) \, \hat b_n (t) .
\end{eqnarray}
This has the effect of replacing the kernel with its causal version and we have
\begin{eqnarray}
K (t,s) \, \theta (t-s) = \sum_n \sqrt{\lambda_n (t)} \, \phi_n (t; t) \phi_n (s ;t)^\ast ,
\end{eqnarray}
where $\theta (\cdot )$ is the Heaviside step function. The support of each $\phi_n (\cdot ; t)$ being contained in $[0,t]$.

\subsection{Complex Trajectories from The Bargmann Representation}

For technical simplicity, we first consider the case where $ \mathfrak{f}_B = \mathbb{C}^\Omega$ with $\Omega$ the finite set of (non-degenerate) frequencies.
The canonical operators are given in the Bargmann representation by
\begin{eqnarray}
\hat{a}_{\omega }^{\ast }\tilde{\Psi}\left( f \right) = f (\omega )^{\ast } \, \tilde{\Psi}\left( f \right) ,\qquad
\hat{a}_{\omega }\tilde{\Psi}\left( f \right) =\frac{\partial }{\partial f _{\omega }^{\ast }}\tilde{\Psi}\left( f\right)
.
\end{eqnarray}
Note that 
\begin{eqnarray}
\hat{a} (g)^\ast \, \tilde{\Psi}\left( f \right) \equiv\langle f|g \rangle \,  \tilde{\Psi}\left( f\right) .
\end{eqnarray}

We shall assume that $g ( \omega )\neq 0$ for each $\omega \in \Omega $. Under the free evolution of the bath Hamiltonian $H_B$, we find that  $\hat{a}_\omega \to e^{-i \omega t} \hat{a}_\omega$ and so
\begin{eqnarray}
\hat{Z}_{t}=  \sum_{\omega } g (\omega ) e^{i\omega
t} \hat{a}_{\omega }^{\ast }\equiv  \hat{a}( g_t)^\ast .
\label{eq:Z}
\end{eqnarray}
with $g_t (\omega ) = e^{i \omega t } g (\omega )$. The eigen-relation (\ref{eq:eigen}) may be written as 
\begin{eqnarray}
\bigg( \hat{Z}_t ^\ast - \zeta _t (f) ^\ast \bigg)   \, |e^f \rangle _B = 0 ,
\label{eq:eigen1}
\end{eqnarray}
where we recall the complex trajectory from Definition \ref{def:complex_traj}
\begin{eqnarray}
\zeta_t (f) = \langle f |  g_t  \rangle = \sum_{\omega } g (\omega )e^{i\omega t} f (\omega )^{\ast } .
\label{eq:zeta}
\end{eqnarray}

We see that, in the Bargmann representation, $\hat{Z} (t)$ is the operator of multiplication by $\zeta_t ( \cdot )$:
\begin{eqnarray}
\hat{Z} (t) \, \tilde{\Psi} (f) = \zeta_t (f) \, \tilde{\Psi }(f) .
\end{eqnarray}

This implies to the more general case where the feature space $\mathfrak{f}_B$ is not finite dimensional.

\begin{definition}
Suppose that the set up is faithful so that the mapping $\zeta : \mathfrak{f}_B \mapsto \mathscr{C}_K ( \mathbb{T} , dt)$ is invertible. Then, for each Bargamann function $\tilde \Psi$, we define the functions $\Psi$ by
\begin{eqnarray}
\Psi (\cdot ) = \tilde{\Psi } \circ \zeta^{-1} .
\end{eqnarray}
Additionally, we can endow $\mathscr{C}_K ( \mathbb{T} , dt)$ with the pull-back pro-measure
\begin{eqnarray}
\mathbb{P} = \tilde{\mathbb{P}} \circ \zeta^{-1} .
\end{eqnarray}
We may extend $\mathbb{P}$ to a probability measure as outlined earlier.
\end{definition}

In detail, for each complex trajectory $\zeta \equiv \zeta (f)$, we have the change of variable $ \Psi (\zeta ) \equiv \tilde \Psi ( f )$.

In this way, $\zeta _{t}=\zeta _{t}\left( \cdot \right) $ is the coordinate for trajectories $\mathscr{C}_K ( \mathbb{T} , dt)$.  in and $\hat{Z} (t)$ may be thought of as the operator corresponding to multiplication by $\zeta _{t} (\cdot ) $.
Indeed, for the case where $\mathbb{T} = \mathbb{R}$, the adjoint operator $\hat{Z} (t)^{\ast }$ may be represented as a functional
differential operator:
\begin{eqnarray}
\hat{Z} (t)^{\ast }\equiv   \int_{-\infty}^\infty d\tau \,K\left( t,\tau \right) \frac{\delta }{\delta \zeta _{\tau } }.
\label{eq:Z_star}
\end{eqnarray}

\begin{remark}
It is easy to see that the correct commutation relations (\ref{eq:Z_CR}) follow from the representation of $\hat Z (t)$ with multiplication by $\zeta_t$ and $\hat Z (t)^\ast$ with \ref{eq:Z_star}.  However, it is instructive to repeat here the argument from \cite{DS97} for a finite dimensional feature space. We have
\begin{eqnarray*}
\hat{Z} (t)^{\ast }\tilde{\Psi}\left( f^{\ast }\right) =\sum_{\omega } g (\omega )^{\ast }e^{-i\omega t}\,\frac{\partial }{\partial f (\omega )^{\ast }}\tilde{\Psi}\left( f^{\ast }\right) .
\end{eqnarray*}
The transform may be inverted to give $ f (\omega )^\ast = \frac{1}{2 \pi} \frac{1}{g(\omega )} \int_{-\infty}^\infty \zeta_\tau (t) e^{-i \omega \tau } d \tau $ and so $\frac{\delta f (\omega )^\ast }{\delta \zeta _{\tau }}
\equiv \frac{1}{2 \pi} \frac{1}{ g (\omega ) }e^{-i\omega \tau }$.
It follows that
\begin{eqnarray*}
\int_{-\infty}^\infty d\tau \,K\left( t,\tau \right) \frac{\delta }{\delta \zeta _{\tau }} = \int_{-\infty}^\infty d\tau \, \sum_{\omega'} |g(\omega' )|^2 e^{-i \omega' (t-\tau )} 
\sum_{\omega }\frac{e^{-i\omega \tau }}{2 \pi \, g (\omega )  }\frac{\partial }{\partial f (\omega )^\ast } \equiv \sum_{\omega } g (\omega )^{\ast }e^{i\omega t}\,\frac{\partial 
}{\partial f( \omega )^{\ast }}.
\end{eqnarray*}
\end{remark}

In the special case here, where $\Omega$ is a finite set of frequencies and $g$ is faithful, we may also construct an inverse kernel.

\begin{proposition}
For a finite dimensional feature space in the faithful case, we can construct the inverse kernel
\begin{eqnarray}
G\left( t,s\right) =\frac{1}{2\pi }\sum_{\omega \in \Omega} \frac{1}{| g (\omega ) |^2}e^{-i\omega (t-s)},
\label{eq:G}
\end{eqnarray}
that is, $\int_{-\infty}^\infty G\left( t,\tau \right) K\left( \tau ,s\right) =\delta \left( t-s\right) $. Moreover, we have the identity
\begin{eqnarray}
\sum_{\omega }| f _{\omega }|^{2}=\int_{-\infty}^\infty\int_{-\infty}^\infty \zeta _{t}(f)^{\ast }G\left( t,s\right) \zeta _{s} (f) \, dtds .
\end{eqnarray}
\end{proposition}


\section{Complex Trajectory Unravellings}
\begin{definition}
Let $\mathcal E$ be a measurable space with measure $\mu$, then a family of maps $\mathcal E  \mapsto \mathfrak{h}_S : \xi \mapsto | \Psi_t (\xi ) \rangle_S $, parametrized by time $t$, is an unravelling of an open quantum system if we have
\begin{eqnarray}
 \langle  j_t (\hat{X}) \rangle = \int_{\mathcal E} \langle \Psi_t ( \xi ) |\hat{X} |\Psi_t (\xi ) \rangle_S \, \mu [ d \xi ]  .
\end{eqnarray}
\end{definition}

We will now construct an unravelling for the open system dynamics encountered in subsection \ref{sec:SB} using the Bargmann representation.


\subsection{Complex Wave Representation for the Open System}
We now consider a system with Hilbert space $\mathfrak{h}_S$ coupled to a the assembly of oscillators which acts as a bath. Any state $| \Psi \rangle \in \mathfrak{h}_S \otimes \mathfrak{H}_B$ can be given a \textit{hybrid complex wave representation} as a map $:\mathbb{C}^\Omega \mapsto \mathfrak{h}_S $ taking $f \mapsto | \tilde{\Psi} (f ) \rangle_S$ where the complex wave representation is used for the bath.

It is convenient to parametrize using complex trajectories $\zeta$ rather than the amplitudes $f$: for $\zeta \in \mathscr{C}_K ( \mathbb{T} , dt) $ we shall use the form $| \Psi (\zeta ) \rangle_S \equiv | \tilde{\Psi} (f_ \zeta ) \rangle_S$. The inner product is then 
\begin{eqnarray}
\langle \Phi |\Psi \rangle =\int_{\mathbb{C}^{\Omega }}
\langle \tilde{\Phi}(f ) | \tilde{\Psi} (f) \rangle_S \, \tilde{\mathbb{P}}\left[ d f \right]
=\int_{\mathscr{C}_K ( \mathbb{T} , dt) }
\langle \Phi (\zeta ) | \Psi (\zeta ) \rangle_S \, \mathbb{P}\left[ d\zeta \right]. 
\end{eqnarray}

Generally speaking, the vector states $| \Psi (\zeta ) \rangle_S $ are not normalized. Instead we have (switching to the distribution over the complex trajectories $t\mapsto \zeta_t$ rather than the complex amplitudes $ \omega \mapsto f(\omega )$)
\begin{eqnarray}
\int_{\mathscr{C}_K ( \mathbb{T} , dt )} \| \Psi (\zeta ) \|^2_S \,\, \mathbb{P} [ d \zeta ] = 1.
\end{eqnarray}
More generally, for a system observable $\hat{X}$, 
\begin{eqnarray}
\langle \Phi | \hat{X} \otimes \hat{I}_B | \Psi \rangle =  
\int_{\mathscr{C}_K ( \mathbb{T} , dt) }
\langle \Phi (\zeta ) | \hat{X} | \Psi (\zeta ) \rangle_S \, \mathbb{P}\left[ d\zeta \right]. 
\label{eq:matrix_elements}
\end{eqnarray}

\begin{proposition}
\label{prop:causal}
Let $| \Psi_t (\zeta ) \rangle_S$ be the hybrid complex wave representation of the the vector $| \Psi _t \rangle = \hat U_t \, | \phi \otimes \mathrm{vac} \rangle$. We have that
\begin{eqnarray}
\frac{\delta}{\delta \zeta _\tau} | \Psi_t (\zeta ) \rangle_S =0, \quad \mathrm{whenever} \quad \tau \notin [0,t].
\end{eqnarray}
\end{proposition}
\begin{proof}
Set $\zeta = \zeta (f)$ and fix $\psi \in \mathfrak{h}_S$, then 
\begin{eqnarray}
\langle \psi | \Psi_t (\zeta ) \rangle_S =
\langle \psi | \tilde{\Psi}_t ( f ) \rangle_S =
\langle \psi \otimes e^f | \vec{T} e^{-i \int_0^t \hat{\Upsilon} (s) ds} | \phi \otimes \mathrm{vac}  \rangle  
=
\langle \psi \otimes e^f | \vec{T} e^{-i \int_0^t [\hat{\Upsilon} - \hat L \otimes \zeta_s ] \, ds} | \phi \otimes \mathrm{vac}  \rangle .
\end{eqnarray}
Here we use the fact that $\langle e^f | = \langle \mathrm{vac} | e^{ \hat {a} (f)}$ and that we may move the operator $e^{ \hat {a} (f)}$ to the right side of $\hat U_t$ with the effect of shifting the creator fields as $\hat{a} (g_s)^\ast \to \hat{a} (g_s)^\ast + \langle f | g_t \rangle \equiv \hat{a} (g_s)^\ast +\zeta_s (f)$. It is then clear that $\langle \psi | \Psi_t (\zeta ) \rangle_S$ depends on $\zeta_s$ for $s\le t$ only.
\end{proof}

Collecting all our results so far together, we arrive at the following conclusion.

\begin{theorem}
\label{thm:unravelling1}
For the model encountered in subsection \ref{sec:SB}, an unravelling is given on the space $\mathscr{C}_K ( \mathbb{T} , dt) $ of complex trajectories with measure $\mathbb{P}$. The unravelling $\zeta \mapsto | \Psi_t (\zeta ) \rangle_S$ satisfies the Di\'{o}i-Strunz differential equation (\ref{eq:nMSSE}) with initial condition $| \Psi_0 (\zeta ) \rangle_S =| \phi \rangle_S$ for all $\zeta$.
\end{theorem}

\begin{proof}
The state $ | \Psi_t \rangle = \hat{U}_t | \phi \otimes \mathrm{vac} \rangle$ satisfies the differential equation
\begin{eqnarray}
\frac{d}{dt} | \Psi_t \rangle  = 
\big( -i \hat{H} \otimes \hat{I}_B + \hat{L} \otimes \hat{Z} (t) -\hat{L}^\ast \otimes \hat{Z} (t)^\ast  \big) \, | \Psi_t \rangle  .
\end{eqnarray}
Writing this in the hybrid complex trajectory representation leads to
\begin{eqnarray}
\frac{d}{d t} | \Psi_t (\zeta ) \rangle_S  =  \bigg( -i \hat{H} +
 \hat{L}  \zeta_t - \hat{L}^\ast \int_{-\infty}^\infty d\tau \,K\left( t,\tau \right) \frac{\delta }{\delta \zeta _\tau } \bigg)  | \Psi _t(\zeta ) \rangle_S ,
\end{eqnarray}
however, by Proposition \ref{prop:causal} we may restrict the integration to get (\ref{eq:nMSSE}). As we have seen, $\mathbb{P}$ is a Gaussian probability measure on the space $\mathscr{C}_K ( \mathbb{T} , dt) $ of admissible complex trajectories, while (\ref{eq:matrix_elements}) ensures that we have an unravelling.
\end{proof}

We may also derive an alternate form.

\begin{theorem}
\label{thm:unravelling2}
The differential equation for the unravelling $\zeta \mapsto | \Psi_t (\zeta ) \rangle_S$ in Theorem \ref{thm:unravelling1} can alternatively be written in the form
\begin{eqnarray}
\frac{d}{d t} | \Psi_t (\zeta ) \rangle_S  =  \big( -i \hat{H} +
 \hat{L}  \zeta_t \big)   | \Psi_t (\zeta ) \rangle_S 
-  \int_{\mathscr{C}_K ( \mathbb{T} , dt) }\, \xi_t^\ast  \,  \hat{L}^\ast| \Psi _t(\zeta + \xi) \rangle_S \,  \mathbb{P} [ d\xi ] ,
\label{eq:nMSSE2}
\end{eqnarray}
again with the initial condition $| \Psi_0 (\zeta ) \rangle_S =| \phi \rangle_S$ for all $\zeta$.
\end{theorem}

\begin{proof}
Let us fix an orthonormal basis $\{ e_n \}$ for $\mathfrak{h}_S$. We take $f \in \mathfrak{f}_B$ to be the one-particle state such that $\zeta_t = \zeta_t (f)$ and we then have that
\begin{eqnarray}
\langle e_n  | \Psi_t (\zeta ) \rangle_S  =  \langle e_n \otimes e^f | \hat{U}_t | \phi \otimes \mathrm{vac}  \rangle =\langle e_n \otimes \mathrm{vac}  | \hat{U}_t^{(f)} | \phi \otimes \mathrm{vac}  \rangle
\end{eqnarray}
where $\hat{U}_t^{(f)} = e^{\hat{a}(f)} \hat{U}_t e^{- \hat{a} (f) } $. We see that $ \frac{d}{dt} \hat{U}_t^{(f)} = \big( -i \hat{\Upsilon}_t + \zeta_t (f) \, \hat{L} \otimes \hat{I}_B \big) \hat{U}_t^{(f)}$, where we recall that $ \zeta_t (f) \equiv \langle f | g_t \rangle$, and so
\begin{eqnarray}
\frac{d}{d t} \langle e_n  | \Psi_t (\zeta ) \rangle_S  =  -i \langle e_n \otimes \mathrm{vac}  | \hat{\Upsilon}_t\, \hat{U}_t^{(f)} | \phi \otimes \mathrm{vac}  \rangle
+\langle e_n   | \hat{L} | \Psi_t (\zeta ) \rangle_S \, \zeta_t .
\end{eqnarray}
It remains to express the first term on the right hand side in a more convenient form. Here we insert a resolution of identity
\begin{eqnarray}
 -i \langle e_n \otimes \mathrm{vac}  | \hat{\Upsilon}_t\, \hat{U}_t^{(f)} | \phi \otimes \mathrm{vac}  \rangle
=-i \sum_m \int_{\mathfrak{f}_B}  \tilde{\mathbb{P} }[ d k] 
\langle e_n \otimes \mathrm{vac}  | \hat{\Upsilon}_t\, | e_m \otimes e^k \rangle
\langle e_m \otimes e^k |\hat{U}_t^{(f)} | \phi \otimes \mathrm{vac}  \rangle
\label{eq:bridge}
\end{eqnarray}
and we see that
\begin{eqnarray}
-i\langle e_n \otimes \mathrm{vac}  | \hat{\Upsilon}_t\, | e_m \otimes e^k \rangle =-i \langle e_n   | \hat{H} | e_m  \rangle_S
- \langle e_n   | \hat{L}^\ast  | e_m  \rangle_S \, \langle g_t | k \rangle .
\end{eqnarray}

The $\hat{H}$-term contribution to (\ref{eq:bridge}) is then
\begin{eqnarray}
-i \sum_m \int_{\mathfrak{f}_B}  \tilde{\mathbb{P} }[ d k] 
\langle e_n  \otimes \mathrm{vac}  | \hat{H} | e_m  \otimes e^k \rangle
\langle e_m \otimes e^k |\hat{U}_t^{(f)} | \phi \otimes \mathrm{vac}  \rangle 
\end{eqnarray}
which reduces to $-i \langle \hat{H} e_n | \tilde{\Psi}_t (f) \rangle_S = -i \langle  e_n | \hat{H}| \Psi_t (\zeta ) \rangle_S$.

The $\hat{L}^\ast$-term contribution to (\ref{eq:bridge}) is then
\begin{eqnarray}
&&
- \sum_m \int_{\mathfrak{f}_B}  \tilde{\mathbb{P} }[ d k] 
\langle e_n   | \hat{L}^\ast  | e_m  \rangle_S  \, \langle g_t | k \rangle
\langle e_m \otimes e^k |\hat{U}_t^{(f)} | \phi \otimes \mathrm{vac}  \rangle \nonumber\\
&&= - \sum_m \int_{\mathfrak{f}_B}  \tilde{\mathbb{P} }[ d k] 
\langle e_n   | \hat{L}^\ast  | e_m  \rangle_S  \, \zeta _t (k )^\ast
\langle e_m \otimes e^{f+k} |\hat{U}_t | \phi \otimes \mathrm{vac} \rangle \nonumber\\
&&= - \sum_m \int_{\mathfrak{f}_B}  \tilde{\mathbb{P} }[ d k] 
\langle e_n   | \hat{L}^\ast  | e_m  \rangle_S  \, \zeta _t (k )^\ast 
\langle e_m  | \tilde{\Psi}_t ( f+k ) \rangle_S \nonumber\\
&&= -  \int_{\mathfrak{f}_B}  \tilde{\mathbb{P} }[ d k]   \zeta _t (k )^\ast 
\langle e_n  | \hat{L}^\ast  \tilde{\Psi}_t ( f+k ) \rangle_S .
\label{eq:last}
\end{eqnarray}
The final step is the change of functional variable from $k \in \mathfrak{f}_B$ to $\xi_t =  \zeta_t (k)$. We note that $\langle f+k |g_t \rangle = \zeta_t (f)+ \zeta_t (k) = \zeta_t + \xi_t$. This results in (\ref{eq:last}) being  $-  \int_{\mathscr{C}_K ( \mathbb{T} , dt)}  \mathbb{P} [ d \xi ]   \xi _t ^\ast 
\langle e_n  | \hat{L}^\ast  \Psi_t ( \zeta + \xi  ) \rangle_S$.

Collecting the terms, and noting that $e_n$ was arbitrary, yields the desired result.
\end{proof}

\subsection{Thermal Baths}
So far we have taken the bath to be in the vacuum state, though the generalization to thermal baths is straightforward and uses the standard Araki-Woods representation\cite{GK}. Here one should replace the bath with two baths $B_1$ and $B_2$ to describe the spontaneous and stimulated effects, respectively. Here one would set
\begin{eqnarray}
\hat{a}_\omega = \sqrt{ n_\omega (T) +1} \, \hat{a}_{\omega,1} \otimes \hat{I}_{B_2} + \sqrt{ n_\omega (T) } \hat{I}_{B_1} \otimes  \hat{a}_{\omega,1}
\end{eqnarray}
where the $[ \hat{a}_{\omega , j} , \hat{a}_{\omega' , j }^\ast ] = \delta_{\omega , \omega '} \delta_{j,j'}$, for $j=1,2$, and $n_\omega (T) = \frac{1}{e^{\omega/kT} -1}$. 
The thermal state then corresponds to the joint vacuum for $B_1$ and $B_2$. 

Set $ g_1 (\omega ) = \sqrt{ n_\omega (T) +1}  \, g( \omega )$ and $ g_2 ( \omega ) = \sqrt{ n_\omega (T)}  \, g (\omega)^\ast $, then from complex amplitudes $f_j ( \omega ) $ we construct complex trajectories $ \zeta _{t,j} $ corresponding to $\sum_{\omega } g _j (\omega ) e^{-i\omega t} f_j ( \omega  )^{\ast } $. The process $\hat{Z} (t)$ is represented as 
\begin{eqnarray}
\hat{Z}_{1} (t) \otimes \hat{I}_{B_2} + \hat{I}_{B_1} \otimes \hat{Z}_{2}(t)^\ast 
\end{eqnarray}
where
$ \hat{Z} _{j} (t)  =\frac{1}{\sqrt{2\pi }}\sum_{\omega ,j }\lambda _{\omega }e^{-i\omega
t} \hat{a} _{\omega ,j}^{\ast } $. 

The appropriate form of the wave representation now involves a pair of complex trajectories $(\zeta_1 , \zeta_2 )$ and the corresponding Schr\"{o}dinger equation will be
\begin{eqnarray}
\frac{\partial}{\partial t} | \Psi_t (\zeta_1 , \zeta_2 ) \rangle_S  =  \bigg( -i\hat{H} +
\hat{L} \, \zeta_{t,1} -\hat{L}^\ast \int_0^t d\tau \,K_1\left( t,\tau \right) \frac{\delta }{\delta \zeta _{\tau , 1} } 
 -\hat{L}^\ast  \zeta_{t,2} + \hat{L} \int_0^t d\tau \,K_2\left( t,\tau \right) \frac{\delta }{\delta \zeta _{\tau , 2} }
\bigg)  | \Psi _t(\zeta_1 , \zeta_2 ) \rangle_S ,
\label{eq:nMSSE_thermal}
\end{eqnarray}
where $K_j$ denote the kernels associated to the $\lambda_j$, $j=1,2$ respectively.

\subsection{The Markovian Case}
\label{sec:markov}
 however this is best handled by the quantum Ito calculus. Here we should describe proceedings using the Hudson-Parthasarathy quantum stochastic differential equation (QSDE)
\begin{eqnarray}
d \hat{U}_t = \bigg( (-i \hat{H} - \frac{\gamma}{2} \, \hat{L}^\ast \hat{L} ) \otimes dt + \sqrt{\gamma} \, \hat{L} \otimes d \hat{A}^\ast_t - \sqrt{\gamma}\hat{L}^\ast \otimes d \hat{A}_t \bigg) \hat{U}_t ,
\end{eqnarray}
where $\hat{A}_t^\ast$ and $\hat{A}_t$ are creation and annihilation processes and the non-trivial component of the quantum Ito table is $ d\hat{A}_t \, d\hat{A}_t^\ast = dt$. as is well-known, the QSDE admits a unique adapted unitary solution for $\hat{H}=\hat{H}^\ast$ and $\hat{L}$ bounded and $\hat{U}_0 = \hat{I}$. This time, note that the bath is described in terms of the Fock space with one-particle space $L^2 (\mathbb{R}_+ , dt )$.

As before, we set
\begin{eqnarray}
\langle \psi | \Psi_t (\zeta ) \rangle_S = \langle \psi \otimes e^f| \hat{U}_t | \phi \otimes \mathrm{vac}  \rangle
\end{eqnarray}
and we readily obtain
\begin{eqnarray}
\frac{d}{dt} | \Psi_t (\zeta ) \rangle_S = \bigg( -i \hat{H} - \frac{\gamma}{2} \hat{L}^\ast \hat{L}  + \hat{L} \zeta_t  \bigg) | \Psi_t (\zeta ) \rangle_S ,
\label{eq:Markov}
\end{eqnarray}
with $ \zeta_t = \sqrt{\gamma} f(t)$. 


\section{Exact Solution: The Jaynes-Cummings Model}
\label{sec:Exact}
We shall now give an exact solution for the state $| \Psi_t (\zeta ) \rangle_S$ for the Jaynes-Cummings model which is the Hamiltonian model (\ref{eq:Ham_SB}) where our system is a two-level atom with Hilbert space $\mathfrak{h}_S \cong \mathbb{C}^2$ spanned by a ground state $| \mathtt{g} \rangle_S$ and an excited state $| \mathtt{e} \rangle_S$, respectively. In this case, we take the system Hamiltonian to be $\hat{H}= \omega_0 | \mathtt{e}\rangle \langle \mathtt{e} |$ and the coupling operator to be the lower operator $\hat{L} = |\mathtt{g} \rangle \langle \mathtt{e}|$.

In this case, we move to the interaction picture with respect to the unperturbed Hamiltonian $\hat{H}_0 = \hat{H} \otimes \hat{I}_B + \hat{I}_S \otimes \hat{H}_B$. In this case, the interaction picture is described by
\begin{eqnarray}
- i \hat{\Upsilon}_t = |\mathtt{g} \rangle \langle \mathtt{e} | \otimes a (g_t)^\ast 
-  |\mathtt{e} \rangle \langle \mathtt{g} | \otimes a (g_t) ,
\end{eqnarray}
where this time $g_t = e^{it (\hat{h}_B - \omega_0) } g$. The unitary $\hat{U}_t$ admits the formal Dyson series expansion
\begin{eqnarray}
\hat{U}_t = \hat{I} + \sum_{n=1}^\infty (-i)^n
\int_{t \ge \tau_n > \tau_{n-1} > \cdots >\tau_1 \ge0} d\tau_n \cdots d\tau_1 \, \hat{\Upsilon}_{\tau_n} \cdots \hat{\Upsilon}_{\tau_1 } .
\end{eqnarray}

\begin{theorem}
The state $| \Psi_t (\zeta ) \rangle_S$ for the Jaynes-Cummings model with initial system state $| \phi \rangle_S$ is given by
\begin{eqnarray}
| \Psi_t (\zeta ) \rangle_S
= | \phi \rangle_S + \langle \mathtt{e} | \phi \rangle_S \bigg( (\lambda (t) -1) | \mathtt{e} \rangle_S
+\int_0^t \zeta_\tau \lambda (\tau) \, d\tau \, | \mathtt{g} \rangle_S \bigg) ,
\label{eq:JC_equation}
\end{eqnarray}
where $\lambda (t) = \sum_{k=0}^\infty (-1)^k I_k (t)$ with $I_0 (t)= 1$ and, for $k \ge 1$,
\begin{eqnarray}
I_k (t) = \int_{t \ge \tau_{2k} > \tau_{2k-1} > \cdots >\tau_1 \ge0} d\tau_{2k} \cdots d\tau_1 \,
K( \tau_{2k} , \tau_{2k-1} ) \cdots K (\tau_2 , \tau_1 ) .
\label{eq:I-terms}
\end{eqnarray}
\end{theorem}

\begin{proof}
For a given state vector $| \psi \rangle \in \mathfrak{h}_S$ and for $\zeta_t = \zeta_t (f)$, we have that $ \langle \psi | \Psi_t (\zeta ) \rangle_S = \langle \psi \otimes e^f | \hat{U}_t | \phi \otimes \mathrm{vac} \rangle$. 

When computing the terms appearing in the Dyson series, we find that $\hat{\Upsilon}_{\tau_n} \cdots \hat{\Upsilon}_{\tau_1 }$ will be greatly simplified by noting that $\hat{L}^2 = \hat{L}^{*2}=0$.
In fact, we have that
\begin{eqnarray}
\cdots \hat{\Upsilon}_{\tau_4}  \hat{\Upsilon}_{\tau_3 } \hat{\Upsilon}_{\tau_2}  \hat{\Upsilon}_{\tau_1 } | \phi \otimes \mathrm{vac} \rangle &\equiv &
\cdots \hat{L}^\ast \hat{L} \hat{L}^\ast \hat{L} | \phi \rangle_S \otimes  \cdots a ( g_{\tau_4} )a ( g_{\tau_3} )^\ast a ( g_{\tau_2} )a ( g_{\tau_1} )^\ast 
 |  \mathrm{vac} \rangle_B .
\end{eqnarray}
Here $(\hat{L}^\ast \hat{L})^k = | \mathtt{e} \rangle \langle \mathtt{e}|$ for each integer $k >0$.

The development of the Dyson series will begin as follows:
\begin{eqnarray}
\langle \psi \otimes e^f | \hat{U}_t | \phi \otimes \mathrm{vac} \rangle &=& \langle \psi  | \phi  \rangle_S +\langle \psi  | \mathtt{g}  \rangle_S\langle \mathtt{e}  | \phi  \rangle_S \int_0^t d\tau_1 \, \langle f | g_{\tau_1}t \rangle  -\langle \psi  | \mathtt{e}  \rangle_S\langle \mathtt{e}  | \phi  \rangle_S \int_0^t d\tau_2 \int_0^{\tau_2} d\tau_1 \, \langle g_{\tau_2} | g_{\tau_1} \rangle + \cdots .
\end{eqnarray}
The general term is easily deduced.

For $n=2k$ (even), we have
\begin{eqnarray}
&& \langle \psi \otimes e^f | \hat{\Upsilon}_{\tau_{2k} } \cdots \hat{\Upsilon}_{\tau_1 }| \phi \otimes \mathrm{vac} \rangle
\nonumber \\ 
&& \qquad = (-1)^k \langle \psi | \mathtt{e} \rangle_S \langle \mathtt{e} | \phi \rangle_S 
\, \langle e^f | a(g_{\tau_{2k}} ) a (g_{\tau{2k-1}})^\ast \cdots
a (g_{\tau_2} ) a (g_{\tau_1} )^\ast | \mathrm{vac} \rangle_B \nonumber \\ 
&& \qquad = (-1)^k \langle \psi | \mathtt{e} \rangle_S \langle \mathtt{e} | \phi \rangle_S 
\, K ( \tau_{2k} , \tau_{2k-1} )  \cdots K  ( \tau_2 ,  \tau_1) ,
\end{eqnarray}
while for $n=2k+1$ (odd)
\begin{eqnarray}
&& \langle \psi \otimes e^f | \hat{\Upsilon}_{\tau_{2k+1} } \cdots \hat{\Upsilon}_{\tau_1 }| \phi \otimes \mathrm{vac} \rangle
\nonumber \\ 
&&  \qquad = (-1)^k \langle \psi | \mathtt{g} \rangle_S \langle \mathtt{e} | \phi \rangle_S 
\, \langle e^f |  a(g_{\tau_{2k+1}} )^\ast a(g_{\tau_{2k}} )\cdots
a (g_{\tau_2} ) a (g_{\tau_1} )^\ast | \mathrm{vac} \rangle_B \nonumber \\ 
&&  \qquad = (-1)^k \langle \psi | \mathtt{g} \rangle_S \langle \mathtt{e} | \phi \rangle_S 
\, \zeta_{\tau_{2k+1}} \, K ( \tau_{2k} , \tau_{2k-1} )  \cdots K  ( \tau_2 ,  \tau_1) .
\end{eqnarray}
Summing the series yields
\begin{eqnarray}
\langle \psi \otimes e^f | \hat{U}_t | \phi \otimes \mathrm{vac} \rangle
= \langle \psi | \phi \rangle_S + \langle \mathtt{e} | \phi \rangle_S \big( \lambda (t) \langle \psi | \mathtt{e} \rangle_S
+\int_0^t \zeta_\tau \lambda (\tau) \, d\tau \, \langle \psi | \mathtt{g} \rangle_S \big) 
\end{eqnarray}
and this gives the result.
\end{proof}

\begin{remark}
We note that we may write $ | \Psi _t (\zeta ) \rangle_S = \hat{\mathtt{U}}_t (\zeta ) \,  | \phi \rangle_S$ where
\begin{eqnarray}
\hat{\mathtt{U}}_t ( \zeta ) = | \mathtt{g} \rangle \langle \mathtt{g} |  +  \lambda (t) \,| \mathtt{e} \rangle \langle \mathtt{e} |
+ \int_0^t \zeta_\tau   \lambda (\tau) d \tau \,| \mathtt{g} \rangle \langle \mathtt{e} | .
\end{eqnarray}
\end{remark}

\begin{remark}
We note that the integrals (\ref{eq:I-terms}) satisfy $ \dot I_{k+1} (t) = \int_0^t K(t, \tau ) I_k ( \tau ) \, d\tau$ and consequently
\begin{eqnarray}
\dot \lambda (t) = - \int_0^t K(t, \tau ) \lambda (\tau) \, d \tau. 
\label{eq:lambda_diff}
\end{eqnarray}
\end{remark}

\begin{remark}
It is worth noting that the solution (\ref{eq:JC_equation}) is indeed a solution to (\ref{eq:nMSSE}) from Theorem \ref{thm:unravelling1} when specified to the Jaynes-Cummings model. Indeed, (\ref{eq:nMSSE}) for this problem reads as
\begin{eqnarray}
\frac{d}{d t} | \Psi_t (\zeta ) \rangle_S  =    \zeta_t \langle \mathtt{e} | \Psi_t (\zeta ) \rangle_S \, | \mathtt{g} \rangle_S
 - \int_{-\infty}^\infty d\tau \,K\left( t,\tau \right) \langle \mathtt{g} |\frac{\delta }{\delta \zeta _\tau }  \Psi _t(\zeta ) \rangle_S \, | \mathtt{e} \rangle_S .
\end{eqnarray}
We likewise see that (\ref{eq:JC_equation}) satisfies
\begin{eqnarray}
\frac{d}{d t} | \Psi_t (\zeta ) \rangle_S  &=&    \zeta_t \langle \mathtt{e} | \Psi_t (\zeta ) \rangle_S \, | \mathtt{g} \rangle_S
 + \dot \lambda (t) \langle \mathtt{e} | \phi  \rangle_S \, | \mathtt{e} \rangle_S \nonumber \\
&=&  \zeta_t \langle \mathtt{e} | \Psi_t (\zeta ) \rangle_S \, | \mathtt{g} \rangle_S
- \int_0^t K(t, \tau ) \lambda (\tau) \, d \tau \, \langle \mathtt{e} | \phi  \rangle_S \, | \mathtt{e} \rangle_S ,
\end{eqnarray}
where we used identity (\ref{eq:lambda_diff}) in the last line. Comparison of the two differential equations shows that the coincide if
\begin{eqnarray}
\langle \mathtt{g} |\frac{\delta }{\delta \zeta _\tau }  \Psi _t(\zeta ) \rangle_S = 
\lambda (\tau) \, \langle \mathtt{e} | \phi  \rangle_S 
.
\end{eqnarray}
Again, for the solution (\ref{eq:nMSSE}) we find that 
\begin{eqnarray}
\langle \mathtt{g} |\frac{\delta }{\delta \zeta _\tau }  \Psi _t(\zeta ) \rangle_S=
\frac{\delta }{\delta \zeta _\tau } \bigg( \int_0^t \zeta_t \lambda (\tau ) d\tau \,  \bigg) \langle \mathtt{e} | \phi  \rangle_S 
=\bigg( \int_0^t     \delta (t- \tau )\lambda (\tau )  \, d \tau  \bigg)    \langle \mathtt{e} | \phi  \rangle_S 
=
\lambda (t) \, \langle \mathtt{e} | \phi  \rangle_S .
\end{eqnarray}
\end{remark}

\begin{remark}
The Markovian case corresponds to $K(t,s) = \gamma \, \delta (t-s)$. Here we find that the integrals (\ref{eq:I-terms}) reduce to $k$-simplicial integals and yield $I_k (t) = \frac{1}{k!} (\gamma /2 )^k $, so that $ \lambda (t) = e^{-\gamma t /2}$ giving an exponential relaxation.

It is instructive to see that this solution is indeed a solution of Markovian equation (\ref{eq:Markov}) for the Jaynes-Cumming model: here $\mathfrak{h}_S =\mathbb{C}^2$ and we set $\hat{H} =0$ and $\hat{L} = | \mathtt{g} \rangle \langle \mathtt{e} |$. The markovian complex trajectory equation is then 
\begin{eqnarray}
\frac{d}{dt} | \Psi_t (\zeta ) \rangle_S = \bigg(  - \frac{\gamma}{2} | \mathtt{e} \rangle \langle \mathtt{e} |  +  \zeta_t | \mathtt{g} \rangle \langle \mathtt{e} | \bigg) | \Psi_t (\zeta ) \rangle_S 
\end{eqnarray}
and setting
\begin{eqnarray}
| \Psi_t (\zeta ) \rangle_S =
\left[
\begin{array}{c}
c_e (t) \\
c_g (t) 
\end{array}
\right] ,
\end{eqnarray}
we find $ \dot{c}_e (t) = \frac{\gamma}{2} \, c_e (t)$ and $\dot{c}_g (t) =\zeta_t \, c_e (t)$
with solutions $c_e (t) = c_e (0) e^{-\gamma t /2}$ and $c_g (t) = c_g (0) + c_e (0) \int_0^t \zeta_\tau e^{-\gamma \tau /2} d\tau$ with 
$ c_e (0) = \langle \mathtt{e} | \phi \rangle_S $ and $ c_g (0) = \langle \mathtt{g} | \phi \rangle_S $. This solution coincides with that given by (\ref{eq:JC_equation}) when we take the exponential relaxation $\lambda (t) = e^{- \gamma t /2}$. 
\end{remark}

\section{Can We Observe Complex Trajectories?}
\label{sec:observe}
Quantum filtering theory is based on the requirement that we measure a commuting family of observable processes $\{ \hat{Y}_k (\cdot ) :  k=1, \cdots, n\}$ parameterized by time. Commutativity is essential because otherwise the most recent measurement demolishes previous ones: this is referred to as the \textit{non-self-demolition property} of the observations. Likewise, an observable can be estimated in terms of these measurements only if it commutes with them - this is the \textit{non-demolition principle}. 

At this stage it is useful to recall the basics of von Neumann algebras. A collection of operators is said to form an algebra if it is closed under addition, scalar multiplication and operator multiplication. It is a *-algebra if, whenever $\hat{X}$ is in the algebra,  $\hat{X}^\ast $ is as well. The commutant of an algebra $\mathfrak{A}$ which is a subalgebra of a larger algebra $\mathfrak{B}$ is the set of all elements in $\mathfrak{B}$ which commute with every element of $\mathfrak{A}$: we denote the commutant by $\mathfrak{A}'$. A von Neumann algebra is then as *-algebra which is its own bicommutant, that is, $\mathfrak{A}''=(\mathfrak{A}')'= \mathfrak{A}$. Von Neumann algebras play the role of $\sigma$-algebras in noncommutative measure theory. Given a collection of operators $\mathfrak{A}_0$, we can construct the smallest von Neummann algebra containing all the elements of the collection - this will be $\mathfrak{A_0}''$ and we refer to this as the von Neumann algebra generated by $\mathfrak{A}_0$.

We may now formulate the quantum filtering problem. The von Neumann algebra $\mathfrak{Y}_t$ generated by the measured observables $\{ \hat{Y}_k (\tau ) :   0 \le \tau \le t, k=1, \cdots, n\}$ will be commutative, by virtue of the self-nondemolition property. 
We say that the measured processes are \textit{essentially classical stochastic processes} meaning that if we fix a state and restrict to just these processes then the classical theory of stochastic processes is sufficient to describe proceedings. 

The nondemolition property then tells us that we should restrict our interests at each time $t$ to those observables that lie in the commutant of measurement algebra up to that times, $\mathfrak{Y}_t'$. If we fix an overall state, then we can construct the conditional expectation from $\mathfrak{Y}_t'$ onto $\mathfrak{Y}_t$. In the case of markovian models, $\mathfrak{Y}_t'$ consists of all observables of the system (in the Heisenberg picture) at time $t$ or later. In such cases, the conditional expectation for an observable $\hat{A}$ of the system ($\hat{A}_t$ in the Heisenberg picture) is denoted by $\pi_t (\hat{A})$ and must, of course, belong to $\mathfrak{Y}_t$. 

As $\mathfrak{Y}_t$ is commutative, we may view it as a classical algebra - effectively treating the observed processes as standard classical stochastic processes. We may construct a conditioned state $| \Psi_t \rangle_S$ (a random variable depending on the observations up to time $t$) such that $\langle \Psi_0 | \pi_t (\hat{A}) | \Psi_0 \rangle_S = 
\langle \Psi_t | \hat{A} | \Psi_t \rangle_S$. In general $| \Psi_t \rangle_S$ satisfies a nonlinear stochastic differential equation but we may also write
\begin{eqnarray}
\langle \chi_0 | \pi_t (\hat{A}) | \chi_0 \rangle_S =
\frac{\langle \chi_t | \hat{A} | \chi _t \rangle_S }{\langle \chi_t | \chi_t \rangle_S }  
\end{eqnarray}
where $| \chi_t \rangle_S $ is not normalized but satisfies a linear stochastic differential equation (the quantum Zakai equation). A typical Zakai equation takes the form
\begin{eqnarray}
d | \chi_t \rangle_S = -( \frac{1}{2}\sum_j \hat{L}_j^\ast \hat{L}_j + i \hat{H} ) \, | \chi_t \rangle_S  \, dt +  \sum_j \hat{L}_j | \chi_t \rangle_S \, dY_j (t)
\end{eqnarray}
where the $\hat{L}_j$ are coupling operators.

\bigskip

Returning to the complex trajectory problem, the first issue we encounter is that the process $\hat{Z} (t)$ not self-adjoint. In principle, we may decomposed $\hat{Z} (t) =\hat{X} (t) + i \hat{Y} (t)$ where $\hat{X}$ and $\hat{Y}$ are self-adjoint processes which we may refer to as the quadrature processes. We note the following commutation relations
\begin{gather}
[ \hat{X}(t)  ,\hat{X} (s) ] = [\hat{Y} (t),\hat{Y} (s)  ] = \frac{i}{2} \, \textrm{Im} \, K(t,s) , \nonumber \\
  [ \hat{X} (t)  , \hat{Y} (s) ] = \frac{1}{2i} \, \textrm{Re} \, K(t,s).
\end{gather}

We may construct the algebra $\mathfrak{Z}_t$ of functions of the operators $\{ \hat{Z} ( \tau ) :  0 \le \tau \le t \}$ and this will indeed be a commutative algebra. However, it is not a *-algebra and in particular will not be a von Neumann algebra. In fact, we can construct the adjoint algebra $\mathfrak{Z}_t^\ast $ and this does not coincide with $\mathfrak{Z}_t$. Indeed, $\mathfrak{Z}_t \cap \mathfrak{Z}_t^\ast$ consists of just multiples of the identity.

The von Neumann algebra generated by $\mathfrak{Z}_t$ is $\mathfrak{W}_t=( \mathfrak{Z}_t \cup \mathfrak{Z}_t^\ast )''$ and this will be noncommutative since the process $Z_t$ does not commute with its adjoint. We may form the von Neumann algebras $\mathfrak{X}_t$ and $\mathfrak{Y}_t$ generated by the quadrature processes but, again, these two processes do not commute. Worse still, neither of the quadrature algebras are commutative - both are \textit{self-demolishing!} This means that we cannot even do filtering on just one of these quadratures. (Exceptionally, the quadrature algebras will be commutative if the kernel was real-valued. This would require that for each frequency $\omega \in \Omega$ we also have $-\omega  \in \Omega$ with $| g(\omega )|= | g (-\omega ) |$.)

\bigskip

In short, $\hat{Z}_t$ is not self-adjoint and therefore not observable!

\begin{acknowledgments}
This work is supported by the ANR project “Estimation et controle des syst\'{e}mes quantiques ouverts” QCOAST Projet ANR-19-CE48-0003, the ANR project QUACO ANR-17-CE40-0007, and the ANR project IGNITION ANR-21-CE47-0015.
\end{acknowledgments}

\end{document}